\def\draftdate{\today}

\documentclass{amsart}
\usepackage{xy}
\usepackage{amssymb}
\usepackage{tikz}

\usetikzlibrary{patterns}
\usetikzlibrary{positioning}
\usetikzlibrary{arrows}
\usepackage{pgfplots}
\usetikzlibrary{calc,fadings,decorations.pathreplacing}

\xyoption{arrow}
\xyoption{matrix}
\xyoption{cmtip}
\xyoption{frame}
\xyoption{line}
\xyoption{all}
\SelectTips{cm}{}

\def\headstrut{\vrule height2.75ex width0ex depth1.5ex}
\def\entrystrut{\vrule height2.25ex width0ex depth0ex}
\def\headlower#1{\lower.25ex\hbox{#1}}

\let\iso\cong

\let\overto\xrightarrow

\def\quickop#1{\expandafter\DeclareMathOperator\csname
#1\endcsname{#1}}
\quickop{Pro}\quickop{Ind}\quickop{Set}\quickop{op}\quickop{colim}\quickop{Hom}\quickop{grp}\quickop{im}\quickop{Map}\quickop{Sd}\quickop{diam}
\quickop{Ex}\quickop{PH}\quickop{St}
\quickop{id}
\quickop{Tel}

\newcommand{\Cech}{\v Cech}

\newcommand{\CMap}{\Map_{\mathrm{SC}}}
\newcommand{\CMapb}{\Map_{\mathrm{SC}*}}
\newcommand{\TMap}{\Map_{\mathrm{Top}}}

\newcommand{\Op}[1]{\mathop{\mathrm{Op}}(#1)}
\newcommand{\Sta}[1]{\mathop{\mathrm{St}}(#1)}
\let\Stv=\Sta

\DeclareMathOperator{\mesh}{Mesh}

\newcommand{\bR}{\mathbb{R}}

\newcommand{\aU}{\mathcal{U}}
\newcommand{\aV}{\mathcal{V}}

\DeclareSymbolFont{emcmbx}{OT1}{cmr}{bx}{n}   
\DeclareMathSymbol{\DDelta}{\mathalpha}{emcmbx}{"01}

\newtheorem{thm}{Theorem}[section]

\newtheorem{lem}[thm]{Lemma}
\newtheorem{prop}[thm]{Proposition}
\theoremstyle{definition}
\newtheorem{defn}[thm]{Definition}
\newtheorem{cons}[thm]{Construction}
\newtheorem{notn}[thm]{Notation}
\newtheorem{alg}[thm]{Algorithm}
\newtheorem{example}[thm]{Example}
\theoremstyle{remark}

\renewcommand{\theenumi}{\roman{enumi}}\renewcommand{\labelenumi}{(\theenumi)}
\renewcommand{\theenumii}{\alph{enumii}}

\newcommand\pgfmathsinandcos[3]{%
  \pgfmathsetmacro#1{sin(#3)}%
  \pgfmathsetmacro#2{cos(#3)}%
}

\newcommand\LatitudePlane[3][current plane]{%
  \pgfmathsinandcos\sinEl\cosEl{#2} 
  \pgfmathsinandcos\sint\cost{#3} 
  \pgfmathsetmacro\yshift{\cosEl*\sint}
  \tikzset{#1/.estyle={cm={\cost,0,0,\cost*\sinEl,(0,\yshift)}}} %
}

\begin{document}

\title{Quantitative Homotopy Theory in Topological Data Analysis}

\author{Andrew J. Blumberg}
\address{Department of Mathematics, University of Texas at Austin,
Austin, TX \ 78712}
\email{blumberg@math.utexas.edu}

\thanks{This research supported in part by DARPA YFA award N66001-10-1-4043}

\author{Michael A. Mandell}
\address{Department of Mathematics, Indiana University,
Bloomington, IN \ 47405}
\email{mmandell@indiana.edu}

\date{\draftdate}


\begin{abstract}
This paper lays the foundations of an approach to applying Gromov's
ideas on quantitative topology to topological data analysis.  We
introduce the ``contiguity complex'', a simplicial complex of maps
between simplicial complexes defined in terms of the combinatorial
notion of contiguity. We generalize the Simplicial Approximation
Theorem to show that the contiguity complex approximates the homotopy
type of the mapping space as we subdivide the domain.  We describe
algorithms for approximating the rate of growth of the components of
the contiguity complex under subdivision of the domain; this procedure
allows us to computationally distinguish spaces with isomorphic
homology but different homotopy types.
\end{abstract}

\maketitle

\section{Introduction}

The topological approach to data analysis is to organize a data set
(which we think of as a finite metric space) into an appropriate
simplicial complex, either the \Cech\ complex, or more usually (for
efficiency reasons) the Rips complex.  Homotopy theoretic invariants
of the simplicial complex then give abstract information about the
data set.  Most work to date has focused on homology, which although
easy to compute is an extremely weak invariant of a complex $Y$.  A
complete set of invariants is given by looking at homotopy classes of
maps from test complexes $X$ into $Y$ (this is called the Yoneda
Lemma~\cite[\S III.2]{MacLane}).  Of course, we cannot hope to have
all this information, and computing homotopy classes of maps is an
intractable problem.  A less intractable compromise is to consider the
homology of the \emph{space} of maps from $X$ to $Y$ for suitable
specific test spaces $X$.  When $X$ is contractible, this just
recovers the homology of $Y$.  But when $X$ is topological nontrivial,
the homology of the mapping space can provide information that is
invisible to the homology of $Y$.  For instance, the homology of the
mapping space from $S^1$ to $Y$ captures information about the
fundamental group of $Y$; this invariant can detect the difference
between spaces $Y$ and $Y'$ with identical homology.  (See
Section~\ref{sec:comp} for detailed computational exploration of such
an example.)

A key benefit of homotopy invariants is their global nature.  On
the other hand, in data analysis, we often also care about feature
scale.  The notion of \emph{quantitative homotopy theory} introduced
by Gromov \cite{Gromov} suggests a way of organizing global homotopy
information of varying scales by looking at maps and homotopies that
have a bounded Lipschitz constant and considering the asymptotics as
the bound increases.  For a data set, instead of looking at Lipschitz
bounds on continuous maps into an associated simplicial complex, it
makes more sense to look at \emph{subdivisions} and simplicial maps.
For abstract simplicial complexes, a simplicial map $X\to Y$ in
particular has Lipschitz constant bounded by 1 for the standard metric (induced
from the embedding of the geometric realization in Euclidean space) on
the geometric realizations $|X|\to |Y|$; a simplicial map from a
subdivision of $X$ has a higher bound on its Lipschitz constant, which depends on
the size of the simplices of the subdivision.   When $Y$ is a \Cech\ or
Rips complex of a data set (i.e., finite metric space), analogous
remarks hold using the ``intrinsic'' metric on $|Y|$ induced from the
metric on the data set.  Considering continuous maps $|X|\to |Y|$ with
Lipschitz constant bounded above by $\lambda$ is 
analogous to considering simplicial maps to $Y$ from subdivisions of
$X$ with mesh size bounded below by $1/\lambda$.  Thus, instead of
studying the subspaces of continuous maps from $|X|$ to $|Y|$ as we
vary the Lipschitz bound, we can study the space of simplicial maps
from subdivisions $X'$ of $X$ to $Y$ as we refine the subdivision.

This paper sets the foundation for an approach to quantitative
homotopy theory in terms of simplicial maps and subdivisions.  As
such, it accomplishes two things.  Firstly, it presents a definition
of a space of simplicial maps between two simplicial complexes
(Definition~\ref{defcmap}). Secondly, it shows that as the mesh size
of the subdivision tends to zero, this space of simplicial maps
approximates the space of continuous maps between the geometric
realizations (Theorem~\ref{thmapprox}).

The first issue of defining the space of simplicial maps makes use of
the a natural notion of ``closeness'' for simplicial maps called
\emph{contiguity}. Simplicial maps $f_{0},f_{1},\dotsc ,f_{n}$ from
$X$ to $Y$ are \emph{mutually contiguous} means that for every simplex
$v_{0},\dotsc,v_{m}$ of $X$, the set $\{f_{i}(v_{j})\}$ of vertices of
$Y$ is a simplex (cf.~\cite[3.5]{Spanier}).  The notion of contiguity is
particularly well-suited to the situation when the target is the Rips
complex of a finite metric space.  In that case, contiguity is a
metric property: the maps are mutually contiguous if and only if any
two points of $\{f_{i}(v_{j})\}$ are within $\epsilon$, the scale
parameter of the Rips complex. 

\begin{defn}\label{defcmap}
For simplicial complexes $X$ and $Y$, define the
\emph{contiguity complex} $\CMap(X,Y)$ to be the simplicial
complex with vertices the simplicial maps and simplices the
collections of mutually contiguous maps.
\end{defn}

As we show in Construction~\ref{conscomp} below, the geometric
realization of $\CMap(X,Y)$ has a canonical natural map 
to the space $\TMap(|X|,|Y|)$ of continuous maps on geometric
realizations (with the compact open topology).  The second issue this
paper addresses is
how good an approximation this 
becomes as we subdivide $X$.  While $\CMap(X,Y)$ is not expected to be
a good approximation of $\TMap(|X|,|Y|)$ (just as the subspace of
Lipschitz constant $\leq 1$ maps from $X$ to $Y$ is not), for successive
refinements $X_{1}$, $X_{2}$, etc., of $X$, we claim that
$\CMap(X_{n},Y)$ approximates $\TMap(|X|,|Y|)$ up to homotopy.
Specifically, we do not expect every homotopy class of map from $|X|$ to
$|Y|$ to be represented by a simplicial map from $X$ to $Y$.  Rather,
the Simplicial Approximation Theorem (e.g., \cite[3.4.8]{Spanier} and
Theorem~\ref{thmSA} below) says that
subdivision of $X$ may be required before a given continuous map $X\to
Y$ is represented up to homotopy by a simplicial map.  We prove the
following analogue for the contiguity mapping spaces; see
Sections~\ref{secsub} and~\ref{secsa} for
details on terminology and Section~\ref{secredux} for the construction
of the map.

\begin{thm}\label{thmapprox}
Let $X$ be a finite simplicial complex and $Y$ any simplicial
complex. Let $X_{n}$ be a sequence of successive subdivisions of $X$
and choose compatible simplicial approximations $X_{n+1}\to
X_{n}$.  If the limit of the mesh size of $X_{n}$ goes to zero, then
$\bigcup \CMap(X_{n},Y)$ is homotopy equivalent to the space
$\TMap(|X|,|Y|)$ of continuous maps from $|X|$ to $|Y|$.  (Here the
union is taken over maps $\CMap(X_{n},Y) \to \CMap(X_{n+1})$ induced
from the simplicial approximations.) 
\end{thm}

Definition~\ref{defcmap} and Theorem~\ref{thmapprox} lay the
foundations for quantitative homotopy theory in topological data
analysis by producing a framework for approximating the homotopy type
of mapping spaces by contiguity complexes.  We discuss in
Section~\ref{sec:comp} a detailed example of the computational
application of our techniques, and in Section~\ref{sec:pers} we sketch
how to integrate this work with the approach of persistent homology.
In future work, we plan to develop the relationship to persistent
homology and homotopy theory of sample spaces.

\section{Review of Simplicial Complexes}

This section reviews the basic definitions of simplicial complex and
geometric realization.  We review the notation for geometric simplices, their
open simplices, and their star neighborhoods.  In this section and the
next two, we follow Spanier \cite{Spanier} in spirit (though not in notation
or details) and substitute references there for proofs whenever possible.

\begin{defn}\label{defSC}
A \emph{simplicial complex} consists of a set of
\emph{vertices} $V$ together with a subset
$S$ of the set of non-empty
finite subsets of $V$, satisfying the 
following properties.
\begin{enumerate}
\item The singleton subsets of $V$ are in $S$.
\item If $\sigma \in S$ and $\tau \subset \sigma$,
$\tau\neq\emptyset$, then $\tau \in S$.
\end{enumerate}
The elements of $S$ are called \emph{simplices}, and $\sigma\in S$ is
called an $n$-simplex if $\#\sigma =n+1$.  If $\sigma \in S$ and $\tau
\subsetneq \sigma$, then $\tau$ is called a \emph{face} of $\sigma$.  A
simplicial complex $X=(V,S)$ is \emph{finite} if $V$ is a finite set
or is \emph{locally finite} if each vertex is contained in at most
finitely many simplices.

A map of simplicial complexes $X=(V,S)$ to $X'=(V',S')$ is a map
$f\colon V\to V'$ such that for each $\sigma \in S$, $f(\sigma)\in S'$.

A subcomplex of $X$ is a simplicial complex $A=(U,R)$ such that
$U\subset V$ and $R\subset S$.
\end{defn}

\begin{example}\label{exdeltan}
The \emph{standard $n$-simplex} is the simplicial complex $\Delta[n]$
with vertex set $\{0,\dotsc,n\}$ and simplex set the set of all
non-empty
subsets of $\{0,\dotsc,n\}$.
\end{example}

A simplicial complex is an abstraction of a certain kind of polytope;
given a simplicial complex, the associated polytope is called the
geometric realization.

\begin{defn}\label{defGRSC}
Let $X=(V,S)$ be a simplicial complex and write $\bR^{V}$ for the real
vector space with basis the elements of $V$, which we topologize with
the standard topology if $V$ is finite, or with the union topology if
$V$ is infinite: a subset $U$ of $\bR^{V}$ is open (or closed) if
and only if its intersection with every finite dimensional subspace is open
(respectively, closed) in the standard topology. The \emph{geometric
realization} $|X|$ is the subspace of $\bR^{V}$ of elements that can
be written as $t_{0}v_{0}+\dotsb +t_{n}v_{n}$ where $t_{i}\geq 0$,
$t_{0}+\dotsb +t_{n}=1$, and $\{v_{0},\dotsc ,v_{n}\}\in S$ is an
$n$-simplex of $X$. 
\end{defn}

The topology on $\bR^{V}$ ensures that $|X|$ has the \emph{union
topology}: a subset $U$ in $|X|$ is open (or closed) if and only if
its intersection with every finite subcomplex is open (respectively,
closed).  It follows that the geometric realization can alternatively
be defined as a quotient space of the disjoint union of 
geometric simplices (see below), one for each simplex of $X$, with faces
identified as per the face relations in $X$.

A function from $V$ into a vector space $E$ induces a \emph{linear
map} $|X|\to E$ as the restriction of the linear map $\bR^{V}\to E$.
For a map of simplicial complexes $X\to X'$, the map $X\to \bR^{V'}$
factors through $|X'|$ and the map $|X|\to |X'|$ is continuous.

Regarding $V$ as an orthonormal set in $\bR^{V}$ endows $|X|$ with the
\emph{standard metric} (via restriction). On finite subcomplexes of
$|X|$, the subspace topology agrees with the metric topology; in
general, the metric topology is coarser than the topology on $|X|$
when $X$ is not a locally finite complex.  In the context of data
sets, $V$ may have a meaningful distance function, inducing a
different (but equivalent, since $V$ is finite) metric on $|X|$, which
we refer to as the \emph{intrinsic metric}.  For definiteness in what
follows, we will always use the standard metric on $|X|$, though for
any equivalent metric, all assertions hold with minor modifications.

We use the following notation for certain subsets of $|X|$.

\begin{notn}\label{notSimp}
For $\sigma=\{v_{0},\dotsc,v_{n}\}$ an $n$-simplex in $X$, write
$|\sigma|$ for the \emph{geometric simplex}, the (closed) subset of
$|X|$ consisting of elements of the form $t_{0}v_{0}+\dotsb
+t_{n}v_{n}$.  Write $\Op{\sigma}$ for the \emph{open simplex}, the
subset of $|\sigma|$ consisting of those elements $t_{0}v_{0}+\dotsb
+t_{n}v_{n}$ with all $t_{i}$ non-zero (and therefore strictly
positive); $\Op{\sigma}$ is an open subset of $|\sigma|$, but not
necessarily an open subset of $|X|$.  For $v$ a vertex in $X$, let
$\Stv{v}$ denote the \emph{star neighborhood} of $v$, the subset of
$|X|$ of elements in $R^{V}$ whose $v$ coordinate is non-zero (and
therefore strictly positive); $\Stv{v}$ is an open subset of $|X|$.
Write $\Sta{\sigma}$ for the \emph{star neighborhood} of $\Op\sigma$,
the subset of $|X|$,
\[
\Sta{\sigma}=\bigcap_{v_{i}\in \sigma}\Stv{v_{i}}=\bigcup_{\tau \supset \sigma}\Op{\tau}
\]
(where $\tau$ ranges over the simplices of $X$ that contain $\sigma$);
$\Sta{\sigma}$ is an open subset of $|X|$ containing $\Op\sigma$.
\end{notn}

For convenience later, we note here the following fact.

\begin{prop}
For a simplicial complex $X=(V,S)$, the collection $\{\Op\sigma\mid
\sigma \in S\}$ is a partition of $|X|$; i.e., an element $x\in |X|$
is in $\Op{\sigma}$ for one and only one simplex $\sigma$. 
\end{prop}

\begin{proof}
For $x\in |X|$ (or indeed in $\bR^{V}$), there is one and only one way
to write $x$ as $t_{0}v_{0}+\dotsb +t_{n}v_{n}$ for distinct 
elements $v_{0},\dotsc,v_{n}$ of $V$ with $t_{0}\dotsb t_{n}\neq 0$.
\end{proof}

We use the following easy observation in the next section.

\begin{lem}\label{lemconvex}
Let $X=(V,S)$ be a simplicial complex and let $C$ be a compact subset
of $|X|$ that is convex as a subset of $\bR^{V}$.  Then $C$ is
contained within a geometric simplex of $|X|$.
\end{lem}

\begin{proof}
By compactness $C$ must be contained in $\bR^{V_{0}}$ for some finite
subset $V_{0}=\{v_{0},\dotsc,v_{n}\}$ of $V$ (with the $v_{i}$'s
distinct).  Making $V_{0}$ smaller if necessary, we can assume without
loss of generality that for each $i=0,\dotsc,n$ some element $x_{i}$
of $C$ has non-zero $v_{i}$ coordinate (which then must be positive).
By convexity of $C$, we must 
have $(x_{0}+x_{1})/2$ in $C$, as this is the midpoint of the line
from $x_{0}$ to $x_{1}$.  Likewise, the point
\[
x=\frac{x_{0}+x_{1}}{2^{n}}+\frac{x_{2}}{2^{n-1}}+\dotsb +\frac{x_{n}}{2}
\]
is in $C$.  Since $x\in C\subset |X|$ has $v_{i}$ coordinate positive
for each $i$ (and all other coordinates zero), we must have that
$\sigma =\{v_{0},\dotsc,v_{n}\}=V_{0}$ is a simplex of $X$.  Since
$C\subset |X|\cap \bR^{V_{0}}=|\sigma|$, we see that $C$ is contained
within a geometric simplex of $|X|$.
\end{proof}

\section{Review of Subdivision}\label{secsub}

This section reviews the notion of a subdivision of simplicial complex
and some of its basic properties.

\begin{defn}\label{defsubdivision}
For a simplicial complex $X=(V,S)$, a \emph{subdivision} is a
simplicial complex $X'=(V',S')$ with $V'\subset |X|$ such that the
linear map $|X'|\to \bR^{V}$ factors as a homeomorphism $|X'|\iso |X|$.
\end{defn}

Lemma~\ref{lemconvex} above implies that every geometric simplex
$|\sigma|$ of $|X'|$ has its image contained within a geometric
simplex $|\tau|$ of $|X|$, for some $\tau$ depending on $\sigma$.  (In
particular, the definition above agrees with \cite[3.3]{Spanier}.)
An easy linear algebra argument \cite[3.3.3]{Spanier} in fact shows that
$\Op\sigma$ lands in $\Op\tau$ for the minimal such simplex $\tau$.
The following theorem is useful for 
identifying and conceptualizing subdivisions; see
\cite[3.3.4]{Spanier} for a proof.

\begin{thm}\label{thmpartition}
Let $X=(V,S)$ and $X'=(V',S')$ be simplicial complexes with $V'\subset
|X|$ and assume that the linear map $|X'|\to
\bR^{V}$ factors through $|X|$.  Then $X'$ is a subdivision of $X$ if
and only if for every simplex $\sigma$ of $X$, the set
$\{\Op{\sigma'}\mid \Op{\sigma'}\subset \Op\sigma\}$ is a finite partition
of $\Op\sigma$.
\end{thm}

In particular, since for a zero simplex $\Op\sigma =|\sigma|$, we have
the following immediate consequence (which is also easy to see on its own).

\begin{prop}\label{propsubverts}
If $X'=(V',S')$ is a subdivision of $X=(V,S)$, then $V'\supset V$.
\end{prop}

Important examples of systematic subdivisions include the following.

\begin{example}
$X$ is a subdivision of itself.
\end{example}

\begin{example}[Barycentric Subdivision]
For an $n$-simplex $\sigma=\{v_{0},\dotsc,v_{n}\}$ of $X=(V,S)$, let
$b_{\sigma}=(v_{0}+\dotsb +v_{n})/(n+1) \in \bR^{V}$ be its
\emph{barycenter}, and let $V'=\{ b_{\sigma}\mid \sigma \in
S\}$ be the set of all barycenters of all simplices.  Let
\[
S'=\{\,\{b_{\sigma_{0}},\dotsc,b_{\sigma_{n}}\}\subset V'\mid
\sigma_{0}\subset \dotsb \subset \sigma_{n}\subset V \}.
\]
Then $X'=(V',S')$ is a subdivision of $X$ \cite[3.3.9]{Spanier} called
the \emph{barycentric subdivision}.
\end{example}

For $X'=(V',S')$ a subdivision of $X=(V,S)$, using the homeomorphism
$|X'|\iso |X|$, we can regard any subdivision of $X'$ as a subdivision of $X$
whose vertex set contains $V'$; likewise, any subdivision of $X$ whose
vertex set contains $V'$ can be regarded as a subdivision of $X'$.
Iterating barycentric subdivision then produces subdivisions with
successively finer partitions of the open simplices of $|X|$.  We use
the standard metric on $|X|$ to compare sizes of simplices of
different subdivisions.

\begin{defn}\label{defmesh}
Let $X'$ be a subdivision of $X$.  The \emph{mesh size} of $X'$,
$\mesh(X')$, is the supremum of the diameters in $|X|$ of images of
the geometric simplices of $|X'|$. 
\end{defn}

Since the standard metric on $|X|$ comes from a norm on $\bR^{V}$, we
have that the mesh size of $X'$ is the supremum of the distances in
$|X|$ between all pairs of vertices of $X'$ that span a $1$-simplex of
$X'$, 
\[
\mesh(X')=\sup \{ d_{|X|}(v,w) | \{v,w\}\in S' \}.
\]
In the case when $X'=X$, the mesh size is $\sqrt{2}$, and in general
the mesh size must be $\leq \sqrt{2}$.  In the case
when $X'$ is the barycentric subdivision of $X$, an easy calculation shows
that the mesh size is bounded above by $1$.  For the $n$-th iterated
barycentric subdivision, the mesh size is bounded above by
$1/\sqrt{2^{n-1}}$, which goes to zero.

\section{Review of Simplicial Approximation}\label{secsa}

This section reviews the definition of a simplicial approximation of
a continuous map between the geometric 
realizations of simplicial complexes.  We then state and prove a
formulation of the Simplicial Approximation Theorem.

\begin{defn}\label{defSA}
Let $X$ and $Y$ be simplicial complexes and let $\phi \colon
|X|\to |Y|$ be a continuous map.  A map of simplicial complexes $f 
\colon X\to Y$ is called a \emph{simplicial approximation} of $\phi$ when
for any  $x\in |X|$ and simplex $\sigma$ in $Y$, $\phi(x)\in \Op\sigma$
implies  $|f|(x)\in |\sigma|$.
\end{defn}

If $\phi$ happens to send a vertex $v$ in $|X|$ to a vertex in $|Y|$,
then $f$ must send $v$ to the same vertex of $Y$.  As a consequence
we get the following proposition.

\begin{prop}
Let $f \colon X\to Y$ be a simplicial approximation of a continuous map
$\phi\colon |X|\to |Y|$.  Let $A$ be a subcomplex of $X$ such that $\phi$
restricted to $|A| \subset |X|$ is the geometric realization of a map
of simplicial complexes $g\colon A\to Y$.  Then $f$ restricted to
$A$ is~$g$. 
\end{prop}

Given $x\in |X|$ and $\sigma$ the unique simplex in $Y$ such that $\phi(x)\in
\Op\sigma$, by convexity of $|\sigma|$ in $\bR^{W}$, the entire line
segment between $\phi(x)$ and $|f|(x)$ is contained in $|\sigma|$.  It
follows that the continuous map $H\colon |X|\times [0,1]\to \bR^{W}$
defined by 
\[
H(x,t)=(1-t)\cdot \phi(x)+t\cdot |f|(x)
\]
factors through $|Y|$.  This then defines a homotopy from $\phi$ to $|f|$,
proving the following proposition.

\begin{prop}\label{propsaho}
Let $f \colon X\to Y$ be a simplicial approximation of a continuous map
$\phi\colon |X|\to |Y|$.  Then $\phi$ and $|f|$ are homotopic.  Moreover,
if $\phi$ and $|f|$ agree on a subset $A$ of $|X|$, then $\phi$ and
$|f|$ are homotopic rel $A$.
\end{prop}

It is not always possible to find a simplicial approximation for a
given continuous map.  The following alternative characterization of a simplicial
approximation is useful in this regard.  The proof is
straight-forward; see \cite[3.4.3]{Spanier}.

\begin{thm}\label{thmaltdefSA}
Let $X=(V,S)$ and $Y=(W,T)$ be simplicial complexes.  A map of vertex
sets $f \colon V\to W$ is a simplicial approximation of a continuous map
$\phi \colon |X|\to |Y|$ if and only if for every vertex $v$ in $X$,
$\phi (\Stv{v})\subset \Stv{f(v)}$. 
\end{thm}

In particular, a continuous map $\phi \colon |X|\to |Y|$ admits a simplicial
approximation if and only if for every vertex $v$ of $X$, there exists
a vertex $w$ of $Y$ such that $\Stv{v}\subset \phi^{-1}(\Stv{w})$.
This leads directly to the following theorem, the Simplicial
Approximation Theorem.  For this, recall that for a metric space $M$,
the \emph{Lebesgue number} of an open cover $\aU=\{U_{\alpha}\}$ of $M$ is a number
$\epsilon>0$ (if one exists) such that whenever a subset $D$ of $M$  has
diameter $\leq \epsilon$, there exists $U_{\alpha}$ in $\aU$ such
that $D\subset U_{\alpha}$.  When $M$ is compact, every open cover has
a Lebesgue number%
\iffalse 
. (Proof: let $\aV$ be the open cover consisting of
all open balls $B_{\delta}(x)$ such that $B_{2\delta}(x)\subset
U_{\alpha}$ for some $\alpha$; choosing a finite subcover
$B_{\delta_{1}}(x_{1}),\dotsc,B_{\delta_{n}}(x_{n})$, and taking
$\epsilon =\min(\delta_{1},\dotsc,\delta_{n})$ will work since any
subset of $X$ of diameter $\epsilon$ is contained within
$B_{\delta_{i}+\epsilon}(x_{i})\subset B_{2\delta_{i}}(x_{i})$ for
some $i$.)
\else
~\cite[27.5]{Munkres}. 
\fi

\begin{thm}[Simplicial Approximation Theorem]\label{thmSA}
Let $X$ and $Y$ be simplicial complexes, and $\phi \colon |X|\to |Y|$
a continuous map.  Assume that the open cover $\{ \phi^{-1}(\Stv{w}) \}$ of
$|X|$ (where $w$ ranges over the vertices of $Y$) has a Lebesgue
number $\lambda >0$.  If $X'$ is a subdivision of $X$ with
$\mesh(X')<\lambda/2$, then the composite map $|X'|\iso |X|\to |Y|$
admits a simplicial approximation.
\end{thm}

\begin{proof}
We apply Theorem~\ref{thmaltdefSA}: for every vertex $v$ of $X'$, the
diameter of $\Stv{v}$ in $|X|$ is $\leq 2\mesh(X')$.
\end{proof}

The following example illustrates the way in which successive
subdivision is required to realize ``larger'' homotopy classes of
maps.

\begin{example}[The simplicial fundamental group]
Consider the model of the circle given by the boundary of the standard
two-simplex $\partial \Delta[2]$, i.e., the simplicial complex
determined by vertices $\{v_0, v_1, v_2\}$ and one simplices $\{(v_0,
v_1), (v_1, v_2), (v_2, v_0)\}$.  The $k$th barycentric subdivision
$\Sd^k \partial \Delta[2]$ has $3 (2^k)$ vertices.  The map $S^1 \to
S^1$ of degree $d$ (i.e., a map which ``winds 
around'' the circle $d$ times) does not admit a simplicial
approximation 
$\Sd^{d-1} \partial \Delta[2] \to \partial \Delta[2]$ but does have a
simplicial approximation
$\Sd^{d} \partial \Delta[2] \to \partial \Delta[2]$.
\end{example}

We need the following additional existence statement for simplicial
approximations implicitly used in the statement of
Theorem~\ref{thmapprox}.

\begin{thm}
If $X'$ is a subdivision of $X$, then the homeomorphism $|X'|\iso |X|$
admits a simplicial approximation.
\end{thm}

\begin{proof}
We use the notation from Definition~\ref{defsubdivision}. For every
vertex $v'$ in $V'$, let $\sigma_{v'}$ be the unique simplex of $X$ with
$v'\in \Op{\sigma_{v'}}$.  Choose $g\colon V'\to V$  
to be any function that takes each $v'$ to some vertex of
$\sigma_{v'}$.  We claim that $g$ is a simplicial map and a simplicial
approximation of the homeomorphism $\phi \colon |X'|\iso |X|$.  To see this, let
$\sigma'=\{v'_{0},\dotsc,v'_{n}\}$ be an $n$-simplex in $X'$ and let
$|\tau|$ be a geometric simplex of $|X|$ that contains the image of
$|\sigma'|$ under $\phi$.  
Since $|\tau|$ contains $v'_{i}$ for all $i$, it must also contain
$|\sigma_{v'_{i}}|$ for all $i$; hence
$\{g(v'_{0}),\dotsc,g(v'_{n})\}\subset \tau$ is a simplex in $X$,
which shows that $g$ is a simplicial map.  We note also that $|g|$
sends $|\sigma'|$ into $|\tau|$.  Next we see that $g$ is a
simplicial approximation of the homeomorphism $|X'|\iso |X|$.  For
$x\in |X'|$, without loss of generality, we can assume $x\in \Op{\sigma'}$, and we note
that by Theorem~\ref{thmpartition}, under the homeomorphism $|X'|\iso
|X|$, $\Op{\sigma '}\subset\Op{\sigma}$
for some simplex $\sigma$ of $X$.  Then $|\sigma|$ contains
$v'_{0},\dotsc,v'_{n}$, and so (as above) $|\sigma|$ contains the
image of $|\sigma'|$ under $|g|$.  In particular $x$ lands in
$\Op\sigma$ under the homeomorphism $|X'|\iso |X|$ and $|g|(x)\in
|\sigma|$. 
\end{proof}

Simplicial approximations, when they exist, need not be unique, but we
do have the following theorem, which gives uniqueness up to contiguity.
Recall from the introduction that simplicial maps
$f_{0},\dotsc,f_{n}\colon X\to Y$ are \emph{mutually contiguous} when for any
simplex $\sigma =\{v_{0},\dotsc,v_{m}\}$ of $X$, the set 
\[
\{f_{0}(v_{0}),\dotsc,f_{0}(v_{m}),
f_{1}(v_{0}),\dotsc,f_{n}(v_{m}) \}
\]
forms a simplex in $Y$.

\begin{thm}\label{thmsamc}
If $f_{0},\dotsc,f_{n}\colon X\to Y$ are all simplicial approximations
to $\phi \colon |X|\to |Y|$, then they are mutually contiguous.
\end{thm}

\begin{proof}
For an $m$-simplex $\sigma =\{v_{0},\dotsc,v_{m}\}$ of $X$, choose a
point $x\in \Op\sigma \subset |X|$, and let $\tau$ be the simplex in
$Y$ such that $\phi(x)\in \Op\tau$.  Then $|f_{i}|(x)\in |\tau|$ so
$f_{i}$ must send $v_{0},\dotsc,v_{m}$ to vertices of $\tau$.
It follows that $\{f_{i}(v_{j})\}\subset \tau$.
\end{proof}

Contiguity is a combinatorial analogue of homotopy; a precise
statement of the relationship is the content of the main theorem of
this paper.  However, in contrast to homotopy of maps, contiguity is
evidently {\em not} an equivalence relation (as it is not transitive).
The following simple example illustrates the relationship between
homotopy classes and contiguity classes~\cite[3.5.5]{Spanier}.

\begin{example}\label{exa:simpfundgp}
Consider again the model of the circle given by the boundary of the
standard two-simplex $\partial \Delta[2]$.  Suppose that the
simplicial maps $f, g \colon \partial \Delta[2]\to \partial \Delta[2]$
both land in proper subcomplexes of $\partial \Delta[2]$.  Then it is
easy to check that $f$ and $g$ must each be contiguous to a constant
map (sending all vertices to the same vertex) and that any pair of
constant maps are contiguous.  Thus, all such simplicial maps form one
``contiguity class'' that represents the null-homotopic map from
$S^1 \to S^1$.  On the other hand, if we take
\begin{gather*}
f(v_{0})=v_{0},\qquad f(v_{1})=v_{1}, \qquad f(v_{2})=v_{0},\\
g(v_{0})=v_{0},\qquad g(v_{1})=v_{2}, \qquad g(v_{2})=v_{0},
\end{gather*}
then $f$ and $g$ are not contiguous because for the $1$-simplex
$\{v_{0},v_{1}\}$ in the domain, the subset
$\{f(v_{0}),f(v_{1}),g(v_{0}),g(v_{1})\}=\{v_{0},v_{1},v_{2}\}$ of
vertices in the codomain is not a simplex. 
The remaining maps $\partial \Delta [2]\to \partial \Delta [2]$ are
all given by permutations of the vertex set.  The even permutations
induce degree $1$ self-maps of $S^{1}$ and the odd permutations induce
degree $-1$ self-maps of $S^{1}$.  All maps of the same degree are
homotopic, but an easy check shows that no two of these maps are
contiguous, and so each vertex permutation is its own contiguity
class.  As we subdivide the domain,
the image of these maps will merge into a single contiguity class for
each degree.
\end{example}

\section{The Product of Simplicial Complexes}

In this section we study the product of simplicial complexes and
construct maps relating its geometric realization to the product on
the geometric realization and to the geometric realization of the
product of subdivisions.

\begin{defn}\label{defprod}
Let $X=(V,S)$ and $Y=(W,T)$ be simplicial complexes.  Define the
product complex $X\boxtimes Y$ to be the simplicial complex $(V\times
W,U)$ where $\{(v_{0},w_{0}),\dotsc,(v_{n},w_{n})\}\in U$ if and only
if $\{v_{0},\dotsc,v_{n}\}\in S$ and $\{w_{0},\dotsc,w_{n}\}\in T$.
(N.B. We do not assume that $v_{0},\dotsc,v_{n}$ are distinct
elements of $V$ or $w_{0},\dotsc,w_{n}$ are distinct
elements of $W$.)
\end{defn}

The product $X\boxtimes Y$ comes with canonical maps of simplicial
complexes $X\boxtimes Y\to X$ and $X\boxtimes Y\to Y$ induced by the
projections $V\times W\to V$ and $V\times W\to W$.  These projection
maps have the following universal property.

\begin{prop}
$X\boxtimes Y$ is the product of $X$ and $Y$ in the category of
simplicial complexes: maps of simplicial complexes $Z\to X\boxtimes Y$
are in bijective correspondence with pairs of maps of simplicial
complexes $Z\to X$, $Z\to Y$.
\end{prop}

The cartesian product of spaces has the corresponding universal
property in the category of spaces, and so the geometric realization
of the projection maps $|X\boxtimes Y|\to |X|$, $|X\boxtimes Y|\to
|Y|$ induce a natural continuous map 
\[
\rho \colon |X\boxtimes Y|\to |X|\times |Y|.
\]
Unless $X$ or $Y$ is discrete (has only zero simplices), this map is
far from a homeomorphism.  For example,
\[
\Delta[m]\boxtimes \Delta[n]\iso \Delta[(m+1)(n+1)-1],
\]
where $\Delta[\bullet]$ denotes a standard simplex
(Example~\ref{exdeltan}). Nonetheless, the projection map $\rho$ does
have a natural section
\[
\varsigma \colon |X|\times |Y|\to |X\boxtimes Y|
\]
induced by the universal bilinear map
\begin{gather*}
\bR^{V}\times \bR^{W}\to \bR^{V}\otimes \bR^{W}=\bR^{V\times W}\\
(s_{0}v_{0}+\dotsb +s_{m}v_{m},t_{0}w_{0}+\dotsb +t_{n}w_{n})\mapsto
\sum s_{i}t_{j}(v_{i},w_{j}).
\end{gather*}
This defines a function $\varsigma$ as above since when
$\{v_{0},\dotsc v_{m}\}$ and $\{w_{0},\dotsc,w_{n}\}$ are simplices of
$X$ and $Y$, $\{(v_{i},w_{j})\}$ is a simplex of $X\boxtimes Y$, and
when $\sum s_{i}=\sum t_{j}=1$, then $\sum s_{i}t_{j}=1$.  In general
the product topology of $|X|\times |Y|$ is coarser than the union
topology, and $\varsigma$ is not a continuous map (cf.~Theorem~2
of~\cite{MilnorReal}).  When $X$ or $Y$ is locally finite, the product
topology on $|X|\times |Y|$ coincides with the union topology, and
$\varsigma$ is a continuous map. 

\begin{thm}
The composite function $\rho \circ \varsigma$ is the identity on $|X|\times
|Y|$.  The identity map on $X\boxtimes Y$ is a simplicial
approximation to the composite $\varsigma \circ \rho$, which is a
continuous map.
\end{thm}

\begin{proof}
An arbitrary element of $|X|\times |Y|$ is of the form
\[
(x,y)=(s_{0}v_{0}+\dotsb +s_{m}v_{m},t_{0}w_{0}+\dotsb +t_{n}w_{n})
\]
where
$\{v_{i}\}$ is an $m$-simplex of $X$, $\{w_{j}\}$ is an $n$-simplex of
$Y$, and $\sum s_{i}=\sum t_{j}=1$, $s_{i} > 0$, $t_{j} > 0$.
This element maps to $\sum 
s_{i}t_{j}(v_{i},w_{j})$ in $|X\boxtimes Y|$.  The map induced by the
projection $X\boxtimes Y\to X$ takes this to the element 
\[
s_{0}(t_{0}+\dotsb +t_{n})v_{0}+\dotsb +s_{m}(t_{0}+\dotsb
+t_{n})v_{m}
=s_{0}v_{0}+\dotsb +s_{m}v_{m}=x
\]
in $|X|$ and likewise the map induced by the projection $X\boxtimes
Y\to Y$ takes this element to $y$ in
$|Y|$.  It follows that $\rho(\varsigma(x,y))=(x,y)$ and therefore
that $\rho \circ \varsigma$ is the identity on
$|X|\times |Y|$. 

On the other hand, an arbitrary element of $|X\boxtimes Y|$ is of the
form 
\[
z=t_{0}(v_{0},w_{0})+\dotsb +t_{n}(v_{n},w_{n})
\]
with $\sum t_{i}=1$ and $t_{i}>0$
(with the vertices $(v_{i},w_{i})$ distinct, though neither the
vertices $v_{i}$ nor the vertices $w_{i}$ need be distinct).  The
composite $\varsigma \circ \rho$ takes $z$ to the element $\sum
t_{i}t_{j}(v_{i},w_{j})$ of $|X\boxtimes Y|$; in particular, since
$|X\boxtimes Y|$ has the union topology, the
composite $\varsigma \circ \rho$ is continuous even when $\varsigma$
is not.  Let $\sigma
=\{(v_{i},w_{j})\}$, a $k$-simplex of $X\boxtimes Y$ for some $k\leq
(n+1)^{2}-1$.  Since for all $i,j$, $t_{i}t_{j} > 0$, we have that
$\varsigma (\rho (z))\in \Op\sigma$.  Since $(v_{i},w_{i})\in
\sigma$ for all $i,j$, it follows that $z\in |\sigma|$.  This shows that the
identity on $X\boxtimes Y$ is a simplicial approximation of the
composite $\varsigma \circ \rho$.
\end{proof}

We also use the following extension of $\varsigma$.

\begin{thm}\label{thmdefzeta}
Let $X'$ be a subdivision of $X$ and $Y'$ a subdivision of $Y$.  The
map $\varsigma$ extends to a continuous map 
\[
\zeta \colon |X'\boxtimes
Y'|\to |X\boxtimes Y|,
\]
which sends every geometric simplex to a set of diameter $\leq
\mesh(X')+\mesh(Y')$.  Moreover, given $a\colon X'\to X$ and $b\colon
Y'\to Y$ simplicial approximations to the homeomorphisms $|X'|\iso
|X|$ and $|Y'|\iso |Y|$, then $a\boxtimes b$ is a simplicial
approximation to~$\zeta$.
\end{thm}

\begin{proof}
Write $X'=(V',S')$, $Y'=(W',T')$, and recall that by definition
$V'\subset \bR^{V}$ and $W'\subset \bR^{W}$.
The composite map 
\[
\bR^{V'}\times \bR^{W'}\to \bR^{V}\times \bR^{W}\to \bR^{V}\otimes \bR^{W}=\bR^{V\times W}
\]
induces a unique bilinear map $\bR^{V'\times W'}\to \bR^{V\times W}$;
specifically, if 
\begin{gather*}
v'=s_{0}v_{0}+\dotsb +s_{m}v_{m}\in \bR^{V}\\
w'=t_{0}w_{0}+\dotsb +t_{n}w_{n}\in \bR^{W}
\end{gather*}
then $(v',w')\mapsto \sum s_{i}t_{j}(v_{i},w_{j})\in \bR^{V\times
W}$.  To see that this restricts to a map $\zeta$ as in the statement,
consider an arbitrary element 
\[
z=r_{0}(v'_{0},w'_{0})+\dotsb+r_{\ell}(v'_{\ell},w'_{\ell})\in
|X'\boxtimes Y'|,
\]
where $r_{i}>0$ and $\sum r_{i}=1$ (with the
vertices $(v'_{i},w'_{i})$ distinct).  Write 
\begin{gather*}
v'_{i}=s^{i}_{0}v^{i}_{0}+\dotsb +s^{i}_{m_{i}}v^{i}_{m_{i}}\in \bR^{V}\\
w'_{i}=t^{i}_{0}w^{i}_{0}+\dotsb +t^{i}_{n_{i}}w^{i}_{n_{i}}\in \bR^{W}
\end{gather*}
with the usual conventions.  Then $z$ goes to 
\[
\sum_{i,j,k} r_{i}s^{i}_{j}t^{i}_{k} (v^{i}_{j},w^{i}_{k}).
\]
Since $\{v'_{i}\mid i=0,\dotsc,\ell\}$ is a simplex in $X'$, we must
have that $\{v^{i}_{j}\mid i=0,\dotsc,\ell,j=0,\dotsc,m_{i} \}$ is a
simplex in $X$ and likewise $\{w^{i}_{k}\}$ is a simplex in $Y$.  It
follows that $\{(v^{i}_{j},w^{i}_{k})\}$ is a simplex in $X\boxtimes
Y$, and since $\sum r_{i}s^{i}_{j}t^{i}_{k}=1$, we see that $z$ lands
in $|X\boxtimes Y|$.  This defines $\zeta \colon |X'\boxtimes Y'|\to
|X\boxtimes Y|$; note that $\zeta$ is a linear map and in particular
continuous.  Since $\zeta$ is a linear map, the diameter of the 
image of a simplex $\{(v'_{i},w'_{i})\}$ is the supremum (over
$i,j$) of the distances 
\[
d(\zeta(v'_{i},w'_{i}),\zeta(v'_{j},w'_{j})),
\]
which are clearly bounded above by $\mesh(X')+\mesh(Y')$.  For the
statement about simplicial approximations, let 
\[
\sigma =\{(v^{i}_{j},w^{i}_{k})\mid i=0,\dotsc,\ell, j=0,\dotsc,m_{i}, k=0,\dotsc n_{i}\}.
\]
We note that in the formula for $\zeta (z)$ above, we have
$r_{i}s^{i}_{j}t^{i}_{k} > 0$ for all $i,j,k$, which puts $\zeta (z)$
in $\Op{\sigma}$.  Since $a(v'_{i})\in \{v^{i}_{0},\dotsc
v^{i}_{m_{i}}\}$ and $b(w'_{i})\in \{w^{i}_{0},\dotsc
w^{i}_{n_{i}}\}$, we have that $a\boxtimes
b(\{(v'_{i},w'_{i})\})\subset \sigma $, and hence that $|a\boxtimes
b|(z)\in |\sigma|$.  This shows that $a\boxtimes b$ is a simplicial
approximation of $\zeta$.
\end{proof}

\section{Mapping Complexes and Products}

We now turn to the contiguity mapping complex $\CMap(X,Y)$ defined in
the introduction and the mapping space $\TMap(|X|,|Y|)$, the space of
continuous maps from $|X|$ to $|Y|$ with the compact open topology.
(Recall that the compact open topology is the smallest topology
generated by the sets $B_{U,K}$ of maps $f \colon X \to Y$ such that
$f(K) \subset U$, for $K$ compact and $U$ open.  When $Y$ is a metric
space this topology is equivalent to the topology of uniform
convergence on compact sets.)
For the statement of Theorem~\ref{thmapprox}, we need a comparison map
$\Gamma \colon |\CMap(X,Y)|\to \TMap(|X|,|Y|)$, which we construct in
this section.  We study the behavior of this comparison map under the
adjunction with the product of topological spaces (when $X$ is locally
finite) and an analogous adjunction in simplicial complexes. We begin
with the construction of the comparison map.  

\begin{cons}\label{conscomp}
Let $X$ and $Y$ be simplicial complexes and let $F$ denote the set of
vertices of $\CMap(X,Y)$, i.e., the set of maps of simplicial
complexes $X\to Y$.  A typical
element $\phi$ of $|\CMap(X,Y)|$ then is of the form
$t_{0}f_{0}+\dotsb t_{n}f_{n}\in \bR^{F}$ where $f_{0},\dotsc,f_{n}$
are mutually contiguous and $\sum t_{i}=1$, $t_{i}>0$.  Define $\Gamma\phi$ to be
the function $|X|\to |Y|$ that takes an element
$x=s_{0}v_{0}+\dotsb +s_{m}v_{m}$ of $|X|$ to the element $\sum
t_{i}s_{j}f_{i}(v_{j})$ of $|Y|$.  We note that $\Gamma \phi$ is a
linear map $|X|\to |Y|$ and so in particular $\Gamma \phi$ is a
continuous map.  This then constructs a function $\Gamma \colon |\CMap(X,Y)|\to
\TMap(|X|,|Y|)$.  To see that $\Gamma$ is continuous, it suffices to
check it on each simplex, where it is obvious from the formula.
\end{cons}

The product of topological spaces and the mapping space fit together
into an adjunction as follows.  For spaces $X,Y,Z$, a continuous map
$\phi \colon X\times Z\to Y$ induces a continuous map $\tilde \phi
\colon Z\to \TMap(X,Y)$ defined by 
\[
(\tilde\phi(z))(x)=\phi(x,z).
\]
A standard fact from topology~\cite[46.11]{Munkres} is that when $X$
is a locally compact Hausdorff space, this defines a bijection between
the set of continuous maps $X\times Z\to Y$ and the set of continuous
maps $Z\to \TMap(X,Y)$.  In the case of concern to us, this
specializes to the following proposition. 

\begin{prop}\label{expts}
Let $X$, $Y$, and $Z$ be simplicial complexes and assume that $X$ is
locally finite.  Continuous maps $|Z|\to \TMap(|X|,|Y|)$ are in
one-to-one correspondence with continuous maps $|X|\times |Z|\to |Y|$.
\end{prop}

We have an analogous relationship between the product of simplicial
complexes of the previous section and the contiguity mapping complex,
but without the locally finite hypothesis.

\begin{thm}\label{expsc}
Let $X$, $Y$, and $Z$ be simplicial complexes.   Maps of simplicial
complexes $Z\to \CMap(X,Y)$ are in one-to-one correspondence with maps
of simplicial complexes $X\boxtimes Z\to Y$.
\end{thm}

\begin{proof}
Write $X=(V_{X},S_{X})$ and similarly for $Y$ and $Z$.  Since a map of
simplicial complexes is in particular a map of vertex sets, using the
exponential law for products and functions of sets, we can identify
both the set $A$ of maps of simplicial complexes $X\to \CMap(Y,Z)$ and
the set $B$ of maps of simplicial complexes $X\boxtimes Y\to Z$ as
subsets of the set $C$ of maps $V_{X}\times V_{Z}\to V_{Y}$.  For any $f\in C$,
for any $w\in V_{Z}$, write $f_{w}\colon V_{X}\to V_{Y}$ for the map $f(-,w)$,
and let $\sigma =\{v_{0},\dotsc,v_{m}\}$ and $\tau
=\{w_{0},\dotsc,w_{n}\}$ denote arbitrary simplices of $X$ and $Z$,
respectively.  Starting with $f\in A$, we must have that
$\{f(v_{i},w_{j})\mid i,j\}$ is a simplex of $Y$ since
$f_{w_{0}},\dotsc,f_{w_{m}}$ are mutually contiguous maps of
simplicial complexes.  Since $\sigma$ and $\tau$ were arbitrary, we
see that $f$ is a simplicial map $X\boxtimes Z\to Y$ and hence $f\in
B$.  On the other hand, if we start with $f\in B$, then we know that
$\{f(v_{i},w_{j})\mid i,j\}$ is a simplex of $Y$ and so in particular
for each $j$, $\{f(v_{i},w_{j})\mid i\}$ is a simplex of $Y$.  Since
$\sigma$ was arbitrary this shows that each $f_{w_{i}}$ is a 
map of simplicial complexes $X\to Y$ and that the maps
$f_{w_{0}},\dotsc,f_{w_{m}}$ are mutually contiguous.  Since $\tau$
was arbitrary, this shows that $f$ is a map of simplicial complexes
$Z\to \CMap(X,Y)$ and hence $f\in A$.
\end{proof}

We can relate the correspondences in Proposition~\ref{expts} and
Theorem~\ref{expsc}.  The following proposition is clear from the
formula for $\varsigma$ in the previous section and the formula for
the comparison map $\Gamma$ in Construction~\ref{conscomp}. 

\begin{prop}\label{propadjcomp}
Let $X$, $Y$, and $Z$ be simplicial complexes with $X$ locally finite.
If $f\colon X\boxtimes Z\to Y$ and  $g\colon Z\to \CMap(X,Y)$
correspond under the adjunction of Theorem~\ref{expsc}, then the
composite maps
\begin{gather*}
|X|\times |Z|\overto{\varsigma}|X\boxtimes Z|\overto{|f|}|Y|,\\
|Z|\overto{|g|} |\CMap(X,Y)|\overto{\Gamma} \TMap(|X|,|Y|)
\end{gather*}
correspond under the adjunction of Proposition~\ref{expts}.
\end{prop}

Proposition~\ref{expts} extends to an exponential
correspondence~\cite[1.2.8]{Spanier} (tensor and cotensor adjunctions),
\[
\TMap(|X|\times |Y|,|Z|)\iso \TMap(|Z|,\TMap(|X|,|Y|))
\]
when $X$ is locally finite.  We also have the corresponding extension
of Theorem~\ref{expsc} (for arbitrary $X$).

\begin{prop}
The correspondence of Theorem~\ref{expsc} is the map on vertex sets of
an isomorphism simplicial complexes $\CMap(X\times Y,Z)\iso \CMap(Z,\CMap(X,Y))$.
\end{prop}

\begin{proof}
Let $\alpha=(f_{0},\dotsc,f_{\ell})$ be a finite set of maps $V_{X}\times V_{Z}\to V_{Y}$.
Then $\alpha$ is a simplex of $\CMap(X\times Y,Z)$ or
$\CMap(Z,\CMap(X,Y))$ precisely when for every simplex
$\sigma=\{v_{0},\dotsc,v_{m}\}$ of $X$ and every simplex
$\tau=\{w_{0},\dotsc,w_{n}\}$ of $Z$, the subset
$\{ f_{i}(v_{j},w_{k})\}$
of $V_{Y}$ is a simplex of $Y$.
\end{proof}

\section{A Reduction of Theorem~\ref{thmapprox}}\label{secredux}

In this section, we begin the proof of Theorem~\ref{thmapprox} and
reduce it to a statement about homotopy groups which we prove in the
next section.  For this section and the next, let $X$ and $Y$ be
simplicial complexes with $X$ finite.  Let $X_{0}=X$ and inductively
choose and fix a subdivision $X_{n+1}$ of $X_{n}$ and a simplicial
approximation $f_{n}\colon X_{n+1}\to X_{n}$ of the given
homeomorphism $|X_{n+1}|\iso |X_{n}|$.  We assume that viewed as
subdivisions of $X$, $\mesh{X_{n}}$ tends to zero. This is the setup
in the hypothesis of Theorem~\ref{thmapprox}.

First note that the maps $f_{n}^{*}\colon \CMap(X_{n},Y)\to \CMap(X_{n+1},Y)$ are (up
to isomorphism) inclusions of subcomplexes: since the map $X_{n+1}\to
X_{n}$ is a surjection on vertex sets, a map of simplicial complexes
$X_{n}\to Y$ is completely determined by its composite $X_{n+1}\to
X_{n}\to Y$.  As a consequence, we see in particular that $\bigcup
|\CMap(X_{n},Y)|$ is a CW complex.  Another closely related CW complex
is the telescope, constructed as follows.

\begin{cons}
For a continuous map $\phi \colon A\to B$ between topological spaces
$A$ and $B$, the \emph{mapping cylinder} $M\phi$ is the space
$(A\times [0,1])\cup_{A}B$ obtained by gluing $A\times
[0,1]$ and $B$ along the inclusion of $A$ in $A \times [0,1]$ as $A\times
\{1\}$ and the map $\phi \colon A\to B$.  (This gluing is
characterized by the universal property given by a description as the
pushout  
\[
\xymatrix{
A \ar[r]^-{\id \times \{1\}} \ar[d]_{\phi} & A \times [0,1] \ar[d] \\
B \ar[r] & M\phi.
}
\]
in the category of spaces.)

For a sequence of continuous
maps $\phi_{0}\colon A_{0}\to A_{1}$, $\phi_{1}\colon A_{1}\to A_{2}$,
\dots, the \emph{telescope} $\Tel \phi_{n}$ is the space
\[
M\phi_{0}\cup_{A_{1}}M\phi_{1}\cup_{A_{2}}M\phi_{2}\cup_{A_{3}}\dotsb,
\]
where $M\phi_{n}$ and $M\phi_{n+1}$ are glued together along the
inclusions of $A_{n}$ in $M \phi_{n}$ and in $M\phi_{n+1}$ (as
$A_{n+1}\times \{0\}$). 
\end{cons}

Let $A_{n}=|\CMap(X_{n},Y)|$, let $\phi_{n}=|f_{n}^{*}|$ (the geometric
realization of the map $f_{n}^{*}\colon \CMap(X_{n},Y)\to \CMap(X_{n+1},Y)$ induced by
$f_{n}$) and let $F=\Tel \phi_{n}$. Collapsing the
intervals, we obtain a map $F\to \bigcup A_{n}$.

\begin{prop}
The map $F=\Tel |f_{n}^{*}|\to \bigcup |\CMap(X_{n},Y)|$ is a homotopy equivalence.
\end{prop}

\begin{proof}
This is a standard argument from homotopy theory that just uses the
fact that each $f_{n}^{*}$ is (up to isomorphism) the inclusion of a
subcomplex.  Write
$F_{0}\subset F_{1}\subset \dotsb \subset F$ for the subspaces 
\begin{align*}
F_{0}&=A_{0},\\
F_{1}&=M\phi_{0},\\
F_{2}&=M\phi_{0}\cup_{A_{1}}M\phi_{1},\\
&\vdots
\end{align*}
Note that there is a canonical
inclusion $a_{n}$ of $A_{n}$ in $F_{n}$ via its inclusion as the back face
of $M\phi_{n}$.  We
inductively construct maps $b_{n}\colon A_{n}\to F_{n}$ that are
homotopic to the maps $a_{n}$ but with the property that 
$b_{n}$ restricts on $A_{n-1}$ to $b_{n-1}\colon A_{n-1}\to
F_{n-1}\subset F_{n}$. We start with $b_{0}=a_{0}=\id\colon A_{0}\to
F_{0}$.  For $n>0$, the fact that
$f_{n-1}^{*}$ is (up to isomorphism) the inclusion of a
subcomplex implies that $M\phi_{n-1}=M|f_{n-1}^{*}|$ is a deformation retract of
$A_{n}\times [0,1]$, q.v.~\cite[3.2.4]{Spanier}.  It is slightly more
convenient to regard 
\[
(A_{n-1}\times [-1,1])\cup_{A_{n-1}}A_{n}\subset A_{n}\times [-1,1]
\]
as a deformation retract.  Choosing a retraction $R$, let 
\[
r\colon
A_{n}\to (A_{n-1}\times [-1,1])\cup_{A_{n-1}}A_{n}
\]
be its
restriction to $A_{n}\times \{-1\}\subset A_{n}\times [-1,1]$.  Let
$b_{n}$ be the map obtained as the composite of $r$ with the map
\[
(A_{n-1}\times [-1,1])\cup_{A_{n-1}}A_{n}=(A_{n-1}\times [-1,0])\cup_{A_{n-1}}M\phi_{n-1}\to F_{n}
\]
that on the $A_{n-1}\times [-1,0]$ runs the homotopy from $b_{n-1}$ to
$a_{n-1}$ and on the $M\phi_{n-1}$ piece is the inclusion
$M\phi_{n-1}\subset F_{n}$.
By construction, the restriction of $b_{n}$ to $A_{n-1}$ is $b_{n-1}$,
and the map $R$ induces a homotopy $H_{n}$ from $b_{n}$ to $a_{n}$.
Putting the maps $b_{n}$ together, we get a map $b\colon \bigcup A_{n}\to F$.
Note that the composite of the collapse map $c\colon F\to \bigcup A_{n}$ and the
homotopy $H_{n}$ restricted $A_{n-1}$ is $c\circ H_{n-1}$ for 
the first half and the constant homotopy on the inclusion for the
second half.  Therefore, reparametrizing the $H_{n}$'s, we can fit
these homotopies together to be a homotopy from the composite map
$c\circ b$ on $\bigcup
A_{n}$ to the identity. Similarly, we can use the deformation
retractions $R$ to construct the homotopy between $b\circ c$ and the
identity on $F$.
\end{proof}

To define a map out of $\bigcup A_{n}$, we need maps out of each
$A_{n}$ that are compatible as $n$ varies.  The purpose of the
telescope is that to define a map out of $F$, we only need maps out
of each $A_{n}$ and homotopies that make them compatible as $n$
varies.

\begin{cons}
We construct the continuous map $\Phi \colon F\to \TMap(|X|,|Y|)$ by
specifying compatible maps on each piece $M\phi_{n}$.  On the
$A_{n}\times [0,1]$ part of $M\phi_{n}$, $\Phi$ does the homotopy (induced by
the homotopy of Proposition~\ref{propsaho}) from the map
\[
A_{n}=|\CMap(X_{n},Y)|\overto{\Gamma}\TMap(|X_{n}|,|Y|)\iso |\TMap(|X|,|Y|)|,
\]
to the map
\begin{multline*}
A_{n}=|\CMap(X_{n},Y)|\overto{\Gamma}\TMap(|X_{n}|,|Y|)\overto{|f_{n}|^{*}}\\
|\TMap(|X_{n+1}|,|Y|)| \iso |\TMap(|X|,|Y|)|
\end{multline*}
where the last homeomorphisms are induced by the subdivision
homeomorphisms $|X_{n}|\iso |X|$ and $|X_{n+1}|\iso |X|$, respectively.
On the $A_{n+1}$ part of $M\phi_{n}$, we do the map
\[
A_{n+1}=|\CMap(X_{n+1},Y)|\overto{\Gamma}\TMap(|X_{n+1}|,|Y|) \iso |\TMap(|X|,|Y|)|.
\]
These pieces then fit together to define a continuous map $\Phi$ on $F$.
\end{cons}

Theorem~\ref{thmapprox} now reduces to showing that the map $\Phi$ is a homotopy
equivalence.   Since $X$ is finite, Milnor~\cite{MilnorCW} shows that
$\TMap(|X|,|Y|)$ has the homotopy type of a CW complex, and the 
Whitehead Theorem~\cite[7.6.24]{Spanier} shows that a weak equivalence
(see below) between spaces of the homotopy type of a CW complex must
be a homotopy equivalence.  Thus, Theorem~\ref{thmapprox} reduces to
the following theorem.

\begin{thm}\label{thmredux}
The map $\Phi$ is a weak equivalence.
\end{thm}

As we explain, the proof of this theorem consists of
Lemmas~\ref{lemzero} and~\ref{lemmain} below. 

For the map $\Phi$ to be a \emph{weak equivalence} means that it induces a
bijection on path components and for each element $x\in F$ and each
$q\geq 1$, the induced map on homotopy groups 
\[
\pi_{q}(F,x)\to \pi_{q}(\TMap(|X|,|Y|),\Phi(x))
\]
is an isomorphism.  (See \cite[7.2]{Spanier} for a definition of the
homotopy groups $\pi_{q}$.)  We have the following first step in this
direction.

\begin{lem}\label{lemzero}
$\Phi$ is a bijection on path components.
\end{lem}

\begin{proof}
The Simplicial Approximation Theorem (Theorem~\ref{thmSA}) and
Proposition~\ref{propsaho} shows that $\Phi$ is surjective on path
components.  Theorem 3.5.6 of \cite{Spanier} shows that $\Phi$ is
injective on path components.
\end{proof}

The map on homotopy groups $\pi_{q}(F,x)\to
\pi_{q}(\TMap(|X|,|Y|),\Phi(x))$ will be an isomorphism for one $x$ in
a path component if and only if it is an isomorphism for all $x$ in
that path component.  For $x\in A_{k}$, we can compute
\[
\pi_{q}(F,x)=\colim_{n\geq k} \pi_{q}(A_{n},\phi^{n}_{k}(x))
\]
where $\phi^{n}_{k}$ is the composite map $\phi_{n-1}\circ \dotsb
\circ \phi_{k}$.  If we write $f^{n}_{k}=f_{k}\circ \cdot \circ
f_{n-1}$ and assume that $x$ is a vertex $g\colon X_{k}\to Y$, then
$\phi^{n}_{k}(g)=g\circ f^{n}_{k}$.   The proof of
Theorem~\ref{thmredux} therefore reduces to the following lemma.

\begin{lem}\label{lemmain}
Let $g\colon X_{k}\to Y$ be a map of simplicial complexes and let $q>0$.
The map 
\[
\colim_{n\geq k} \pi_{q}(|\CMap(X_{n},Y)|,g\circ f^{n}_{k})\to 
\pi_{q}(\TMap(|X_{k}|,|Y|),|g|)
\]
is a bijection.
\end{lem}

The proof of this lemma occupies the next section.

\section{The Proof of Lemma~\ref{lemmain}}

This section contains the proof of Lemma~\ref{lemmain}.  Notation and
assumptions are as in the previous section.  For fixed $q\geq 1$, $k\geq 0$, and $g\colon
X_{k}\to Y$, we need to prove injectivity and surjectivity of the map
\[
\colim_{n\geq k} \pi_{q}(|\CMap(X_{n},Y)|,g\circ f^{n}_{k})\to 
\pi_{q}(\TMap(|X_{k}|,|Y|),|g|).
\]
Since the simplicial complexes $X_{n}$ for $n<k$ play no role, by
renumbering, we assume without loss of generality that $k=0$.  

We begin by proving surjectivity.  Let $\psi\colon S^{q}\to
\TMap(|X|,|Y|)$ be a continuous map which sends the chosen base point
of $S^{q}$ to $|g|$.  Choose and fix a simplicial complex $Z$ with a
homeomorphism from its geometric realization to $S^{q}$ having the
base point at a vertex $z$, e.g., $Z=\partial \Delta[q+1]$.  We then
get a map $\alpha \colon |X|\times |Z|\to |Y|$ which restricted to
$|X|\times \{z\}$ is $|g|$.  Since $|X|\times |Z|$ is compact, giving
it the metric given by the $\max$ of the standard metrics on $|X|$ and
$|Z|$, we can choose a Lebesgue number $\lambda >0$ of the open cover
$\{\alpha^{-1}(\Stv{w})\}$ (where $w$ ranges over the vertices of
$Y$).  Choose $n$ large enough so that $\mesh(X_{n}) < \lambda/4$ and
choose a subdivision $Z'$ of $Z$ so that $\mesh(Z') < \lambda/4$.
Then the map
\[
\zeta \colon |X_{n}\boxtimes Z'|\to |X\boxtimes Z|
\]
of Theorem~\ref{thmdefzeta} sends each simplex to a set of diameter $<
\lambda/2$, and so the composite map
\[
|X_{n}\boxtimes Z'|\overto{\zeta} |X\boxtimes Z|\overto{\rho}|X|\times |Y|
\]
sends each simplex to a set of diameter $< \lambda/2$.  The latter map
$\rho \circ \zeta$ therefore sends the star of each vertex to a set
of diameter $< \lambda$.  By Theorem~\ref{thmaltdefSA}, the composite
map 
\[
\alpha \circ \rho \circ \zeta \colon |X_{n}\boxtimes Z'|\to |Y|
\]
admits a simplicial approximation $a\colon X_{n}\boxtimes Z'\to Y$.
Let $b\colon Z'\to \CMap(X_{n},Y)$ be the adjoint map under the
correspondence of Theorem~\ref{expsc}.  On geometric realization, the
composite map
\[
\phi \colon S^{q}\iso |Z'|\to |\CMap(X_{n},Y)|\to
\TMap(|X_{n}|,|Y|)\iso \TMap(|X|,|Y|)
\]
corresponds to the map
\[
|X_{n}|\times |Z'|\overto{\varsigma}|X_{n}\boxtimes Z'|\overto{|a|}|Y|.
\]
Since the composite
\[
|X|\times S^{q}\iso |X_{n}|\times |Z'|\overto{\varsigma}
|X_{n}\boxtimes Z'|\overto{|a|}{|Y|}
\]
is homotopic to $\alpha$ via the homotopy of 
Proposition~\ref{propsaho}, it follows that $\phi$ is homotopic to the
original map $\psi$.  We are not quite finished because $|b|$ is not a
based map; however, $b(z)$ and $g\circ f^{n}_{0}$ are both simplicial
approximations of the map
\[
|X_{n}|\iso |X|\overto{g}|Y|
\]
and therefore span a $1$-simplex in $\CMap(X_{n},Y)$ by
Theorem~\ref{thmsamc}.  The change of base point of $|b|$ along this
path produces an element of $\pi_{q}(|\CMap(X_{n},Y)|,g\circ
f^{n}_{0})$ sent by the homomorphism to the original element of
$\pi_{q}(\TMap(|X|,|Y|),|g|)$ represented by $\psi$.

We now prove injectivity.  Given elements 
\[
[\psi], [\psi'] \in
\colim \pi_{q}(|\CMap(X_{n},Y)|,g\circ f^{n}_{0}),
\]
we can choose $k$ so that both are represented by based maps $S^{q}\to
|\CMap(X_{k},Y)|$, and then renumbering, we can assume without loss of
generality that $k=0$.  Using the Simplicial Approximation Theorem, we
can choose a simplicial complex $Z$ homeomorphic to $S^{q}$ with base
point at a vertex $z$, and based maps $p,p'\colon Z\to \CMap(X,Y)$ of
simplicial complexes representing $[\psi]$ and $[\psi']$,
respectively.  We assume that the composite maps
\[
\Gamma \circ |p|, \Gamma \circ |p'|\colon S^{q}=|Z|\to |\CMap(X,Y)|\to \TMap(|X|,|Y|)
\]
both represent the same element of $\pi_{q}(\TMap(|X|,|Y|),|g|)$, and
we need to see that $|p|$ and $|p'|$ eventually represent the same
element in the colimit.  Choose a homotopy in $\TMap(|X|,|Y|)$, and
let $H$ denote the adjoint map
\[
|X|\times S^{q}\times [0,1]\to |Y|.
\]
Arguing as above, for some $n$ and some refinement $Z'$ of $Z$ and
$I'$ of $[0,1]$, we can construct a map of simplicial complexes
\[
h\colon X_{n}\boxtimes Z'\boxtimes I'\to Y
\]
giving a simplicial approximation of the map
\[
|X_{n}\boxtimes Z'\boxtimes I'|\to |X_{n}|\times |Z'|\times |I'|
\iso |X|\times S^{q}\times [0,1]\overto{H}|Y|.
\]
In particular the composite map
\[
|X|\times S^{q}\times [0,1]\iso
|X_{n}|\times |Z'|\times |I'|\overto{\varsigma} |X_{n}\boxtimes Z'\boxtimes I'|
\overto{|h|}|Y|.
\]
is homotopy equivalent to $H$.  We obtain an adjoint map 
\[
Z'\boxtimes I'\to \CMap(X_{n},Y)
\]
and on geometric realization a map
\[
\eta \colon S^{q}\times [0,1]\iso |Z'|\times |I'|\to |Z'\boxtimes I'|\to |\CMap(X_{n},Y)|.
\]
The restriction $\eta_{0}$ of $\eta$ to $S^{q}\times \{0\}$ and
$(f^{n}_{0})^{*}(p)$ are both adjoint to simplicial approximations of
the same map 
\[
|X_{n}\boxtimes Z'|\overto{\rho}
|X_{n}|\times S^{q}\iso |X|\times S^{q}\to |Y|.
\]
The simplicial approximations are contiguous, and so the maps
$\eta_{0}$ and $(f^{n}_{0})^{*}(p)$ are homotopic maps $S^{q}\to
|\CMap(X_{n},Y)|$.  Likewise the restriction $\eta_{1}$ of $\eta$ to
$S^{q}\times \{1\}$ and $(f^{n}_{0})^{*}(p')$ are homotopic maps
$S^{q}\to |\CMap(X_{n},Y)|$.  Extending $\eta$ with these new
homotopies, we get a homotopy
\[
\eta \colon S^{q}\times [-1,2] \to |\CMap(X_{n},Y)|.
\]
from $|(f^{n}_{0})^{*}(p)|$ to $|(f^{n}_{0})^{*}(p')|$.  Again, we are not
quite done, because we need to worry about the base point. Restricting
to the base point of $S^{q}$, the map 
\[
\alpha \colon [-1,2]\to |\CMap(X_{n},Y)|
\]
defines a loop based at $g\circ f^{n}_{0}$.  The argument is completed
by showing that $\alpha$ is contractible.  By construction, the
restriction to $[0,1]$ is adjoint to the restriction of $h$, 
\[
X_{n}\boxtimes I'\to Y
\]
which is a simplicial approximation to the constant homotopy on the
map 
\[
|X_{n}|\iso |X|\overto{|g|}|Y|.
\]
The composite 
\[
X_{n}\boxtimes I'\to X_{n}\overto{g\circ f^{n}_{0}}Y
\]
is also a simplicial approximation to the same map; these simplicial
maps are contiguous, and so their adjoints are homotopic maps from
$[0,1]$ to $|\CMap(X_{n},Y)|$.  On $\{0\}$ and $\{1\}$, this
homotopy is the restriction of $\alpha$ to $[-1,0]$ and $[1,2]$,
respectively, and this gives a contraction of the loop $\alpha$.

\section{The contiguity complex and topological data
analysis}\label{sec:comp}

\renewcommand{\theenumi}{\arabic{enumi}}\renewcommand{\labelenumi}{\theenumi.}
\renewcommand{\theenumii}{\alph{enumii}}

As we indicated in the introduction, much of our interest in the
contiguity complex comes from its application to topological data
analysis.  When $Y$ is a simplicial complex extracted from a data set
$M$, the homology of the mapping spaces $\TMap(|X|,|Y|)$ provides
refined invariants of $Y$ for suitable test spaces $X$.  This homology
itself is difficult to compute, but the work of the preceding sections
justifies studying approximations given as the homology of the
contiguity complexes $\CMap(X_n, Y)$ for simplicial subdivisions $X_n$
of $X$.

In this section, we describe a computational application of the
contiguity complex, using the circle as a test space.  Specifically,
we present numerical experiments that illustrate how studying the
growth rate of the number of contiguity classes of maps from
successive subdivisions of $S^1$ to two simplicial complexes with
identical homology can distinguish them.  This example demonstrates
both the power of our methodology as well as the some of the
difficulties imposed by the computational complexity of applying it.

Let $T$ denote the simplicial complex with geometric realization
homeomorphic to the 2-torus (see Figure~\ref{torusfiggeom}),
represented by the standard simplicial complex obtained by identifying
the edges of a rectangle.  This complex has 9 vertices, 27
1-simplices, and 18 2-simplices; see Figure~\ref{torusfig} below.  Let
$P$ denote the simplicial complex with geometric realization homotopy
equivalent to a ``doubly pinched sphere'' (see
Figure~\ref{pinchfiggeom}), formed by taking the barycentric
subdivision of $\partial \Delta[3]$, and then adding two 1-simplices
to connect the barycenters of two distinct pairs of $2$-simplices.
This complex has 14 vertices, 38 1-simplices, and 24 2-simplices; see
Figure~\ref{pinchfig} below.  (The geometric realization of this
complex is actually homeomorphic to a sphere with two disjoint
internal chords.)

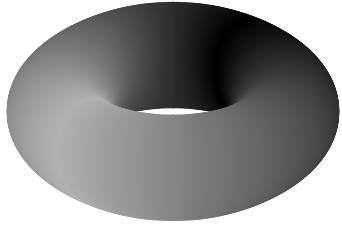
\begin{figure}
\begin{tikzpicture}
\selectcolormodel{gray}
    \foreach \x in {90,89,...,-90} { 

    \pgfmathsetmacro\elrad{20*max(cos(\x),.1)}

    \pgfmathsetmacro\ltint{1.7*abs(\x-45)/180}
    \pgfmathsetmacro\rtint{1.7*(1-abs(\x+45)/180)}
    \definecolor{currentcolor}{rgb}{\ltint, 0, \ltint}

    \draw[color=currentcolor,fill=currentcolor] 
        (xyz polar cs:angle=\x,y radius=.75,x radius=1.5) 
        ellipse (\elrad pt and 20pt);

    \definecolor{currentcolor}{rgb}{\rtint, 0, \rtint}

    \draw[color=currentcolor,fill=currentcolor] 
        (xyz polar cs:angle=180-\x,radius=.75,x radius=1.5) 
        ellipse (\elrad pt and 20pt);
    }
\end{tikzpicture}
\caption{Torus; homeomorphic to the geometric realization of the
        simplicial complex $T$.}\label{torusfiggeom}
\end{figure}

\begin{figure}
\begin{tikzpicture}
[scale=2, vertices/.style={draw, fill=black, circle, inner sep=0.5pt}]
\node[vertices] (a) at (0,0) {};
\node[vertices] (b) at (1,0) {};
\node[vertices] (c) at (2,0) {};
\node[vertices] (d) at (3,0) {};
\node[vertices] (aa) at (0,1) {};
\node[vertices] (bb) at (1,1) {};
\node[vertices] (cc) at (2,1) {};
\node[vertices] (dd) at (3,1) {};
\node[vertices] (aaa) at (0,2) {};
\node[vertices] (bbb) at (1,2) {};
\node[vertices] (ccc) at (2,2) {};
\node[vertices] (ddd) at (3,2) {};
\node[vertices] (aaaa) at (0,3) {};
\node[vertices] (bbbb) at (1,3) {};
\node[vertices] (cccc) at (2,3) {};
\node[vertices] (dddd) at (3,3) {};

\foreach \to/\from in {a/b, b/c, c/d, aa/bb, bb/cc, cc/dd, aaa/bbb,
bbb/ccc, ccc/ddd, aaaa/bbbb, bbbb/cccc, cccc/dddd, a/aa, aa/aaa,
aaa/aaaa, b/bb, bb/bbb, bbb/bbbb, c/cc, cc/ccc, ccc/cccc, d/dd,
dd/ddd, ddd/dddd, aa/b, bb/c, cc/d, aaa/bb, bbb/cc, ccc/dd, aaaa/bbb,
bbbb/ccc, cccc/ddd}
\draw [-] (\to)--(\from);

\end{tikzpicture}
\caption{The simplicial complex $T$, where the top is identified with
the bottom and the left-hand side identified with the right-hand side.}\label{torusfig}
\end{figure}
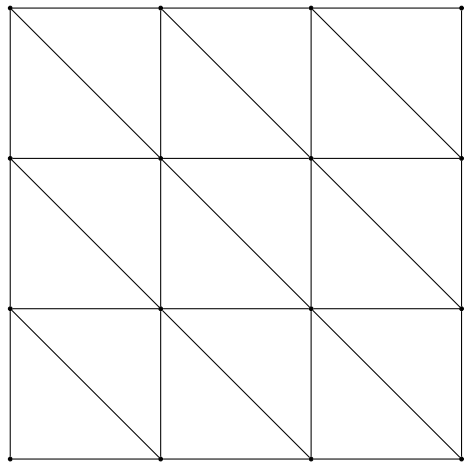

\begin{figure}

\begin{tikzpicture}[scale=1,every node/.style={minimum size=1cm}]
        
        \def\R{4} 
        
        \def\angEl{25} 
        \def\angAz{-100} 
        \def\angPhiOne{-50} 
        \def\angPhiTwo{-35} 
        \def\angBeta{30} 
        
        \LatitudePlane[equator]{\angEl}{0}
        \fill[ball color=white!10] (0,0) circle (\R); 
       
        \draw [dashed, <->] (-1.5,1) -- (-1.5,-1.5);
        \draw [dashed, <->] (1.5,1) -- (1.5,-1.5);
 

    
\end{tikzpicture}

\caption{Doubly pinched sphere, with dashed lines indicating identified points; homotopy equivalent to the geometric realization of the simplicial complex $P$.}\label{pinchfiggeom}
\end{figure}
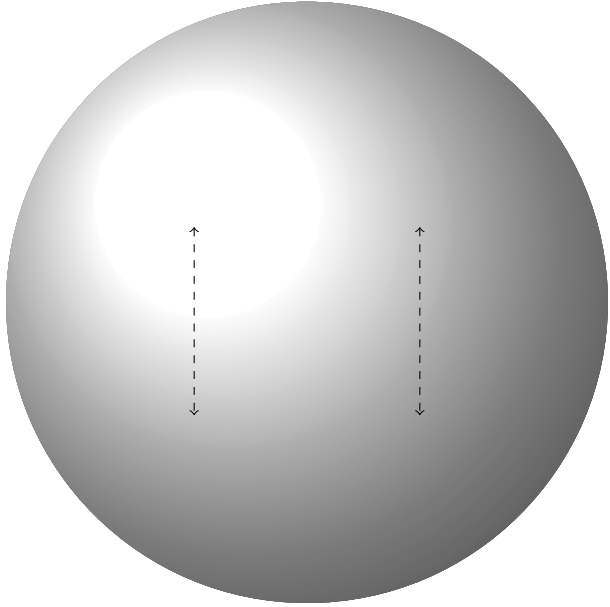

\begin{figure}
\begin{tikzpicture}
[scale=2, vertices/.style={draw, fill=black, circle, inner sep=0.5pt}]
\node[vertices] (a) at (0.2,-0.2) {};
\node[vertices] (ab) at (0.5, -0.1) {};
\node[vertices] (b) at (0.8,0) {};
\node[vertices] (bc) at (0.3, 0.1) {};
\node[vertices] (c) at (-0.2,0.2) {};
\node[vertices] (d) at (-0.8,0) {};
\node[vertices] (e) at (0,0.8) {};
\node[vertices] (be) at (0.4,0.4) {};
\node[vertices] (ae) at (0.1, 0.3) {};
\node[vertices] (de) at (-0.4, 0.4) {};
\node[vertices] (ad) at (-0.3, -0.1) {};
\node[vertices] (abe) at (1.0/3, 0.6/3) {};
\foreach \to/\from in {a/b,a/d,a/e,b/e,d/e, abe/a, abe/b, abe/e,
abe/ab, abe/be, abe/ae, c/e, c/a, c/d, c/ae, c/de, c/ad}
\draw [-] (\to)--(\from);

\draw [bend left=45, ultra thick, dashed] (c)--(abe);
\end{tikzpicture}
\caption{The simplicial complex $P$, with only one of the extra connecting
edges shown.}\label{pinchfig}
\end{figure}
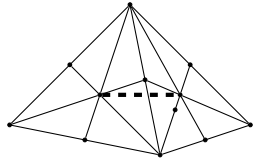

A straightforward computation shows that the simplicial complexes $P$
and $T$ have isomorphic homology; computing the simplicial homology
cannot distinguish these spaces.  However, $P$ and $T$ are not
homotopy equivalent.  To distinguish them, we will study the rank of
the simplicial complex of based maps out of subdivisions of a circle.
(For $X$ and $Y$ simplicial complexes, each having a pre-chosen base
point vertex, we let $\CMapb(X,Y)$ denote the subcomplex of
$\CMap(X,Y)$ consisting of those maps that send the chosen base point
of $X$ to the chosen base point of $Y$.)  In our case, we consider the
following family of subdivisions of a the usual simplicial model of
the circle $\partial \Delta[2]$: let $S^1_k$ denote the simplicial
complex with vertices $\{v_0, v_1, \ldots, v_{k-1}\}$ and edges $(v_i,
v_{i+1})$ for $0 \leq i < k-1$ along with $(v_{k-1}, v_0)$.  We
compute the ranks of $H_0 \CMapb(S^1_k, P)$ and $H_0 \CMapb(S^1_k,
T)$.  

Computing these ranks exactly is computationally intractable for even
relatively small $k$.  A simplicial map from $S^1_k$ to any simplicial
complex $X$ is given by a cycle in the graph specified by the
1-skeleton of $X$, and as $k$ grows the search space of possible
cycles to check grows exponentially.  Moreover, since contiguity is
not an equivalence relation on simplicial maps, we have to compute the
transitive closure of the contiguity relation.  Instead, we use a
randomized algorithm.

First, we note that given two simplicial maps $f$ and $g$, determining
if $f$ and $g$ are contiguous can be done very efficiently.  For one
thing, it suffices to check the contiguity property only on certain
simplices:

\begin{lem}\label{lem:maxsimp}
Let $f,g \colon X \to Y$ be simplicial maps.  Then $f$ and $g$ are
contiguous if and only if $f$ and $g$ satisfy the contiguity property
for every maximal simplex of $X$ (simplex which is not the face of any
other simplex).
\end{lem}

If we are studying a fixed simplicial complex, the maximal simplices
can be precomputed.  Moreover, when the dimension of
$X$ is small, checking the contiguity property for a
given simplex of $X$ can be reduced to a table lookup (where the table
required is of size $v^{2(d+1)}$, for $v$ the number of vertices in $Y$ and
$d$ the dimension of a maximal simplex in $X$).

Next, note that $f$ and $g$ are in the same contiguity class if and
only if there exists a path of simplicial maps $f_0, f_1, \ldots,
f_\ell$ such that $f_0 = f$, $f_\ell = g$, and $f_i$ is contiguous to
$f_{i+1}$ for all $0 \leq i < \ell$.   We have the following randomized
``local search'' algorithm to determine if $f$ and $g$ are in the same
contiguity class.  The idea is to do a random walk on the space of
simplicial maps, starting at $f$ and mostly taking ``greedy'' steps
that get us closer to $g$.  For this, define 
\[
d(f,g) = \sum_{v \in X} d_Y (f(v), g(v)),
\]
where $d_Y$ is the graph distance on the 1-skeleton of $Y$.  In the
case when $X$ is the circle $S^{1}_{k}$ and we are considering based
maps, we can decompose the step of going from $f_{i}$ to $f_{i+1}$
into a sequence of mutually contiguous maps $f_{i}=f_{i,0},
f_{i,1},\dotsc, f_{i,k-1}=f_{i+1}$ where each $f_{i,j+1}$ differs
from $f_{i,j}$ only in a single vertex; specifically, $f_{i,j}$ agrees
with $f_{i+1}$ on vertices $0,\dotsc,j$ and with $f_{i}$ on vertices
$j+1,\dotsc,k-1$.  Thus, for any $f$ and $g$ in the same contiguity
class, it is possible to find a path of contiguous simplicial maps
from $f$ to $g$, where at each step we change a single vertex.  This
justifies the following randomized algorithm that checks whether maps
$f$ and $g$ belong to the same contiguity class.

\begin{alg}\label{alg:contclass}
Fix a parameter $\kappa \in [0,1]$ --- typically values of $\kappa$
are around $0.1$.  (The value of $\kappa$ controls the frequency with
which we take completely random steps as opposed to ``greedy'' steps
that improve the distance to the target map.)  Fix a maximum number of
iterations $M$.
\begin{enumerate}
\setlength{\itemindent}{-2em}
\item Set $f_0 = f$, and $i = 0$.
\item If $f_i = g$, return 1.
\item If $i \geq M$, return 0.
\item Set $i = i + 1$.  
\item Construct $f_i$ as follows:
\begin{enumerate}
\setlength{\itemindent}{-2em}
\item Choose a uniform random vertex $v$ of $X$.
\item Build a candidate map $\tilde{f}$ from $f_{i-1}$ by altering the
value at $v$, uniformly selecting
amongst possible values for $\tilde f(v)$ such that $\tilde f(v) \neq f_{i-1}(v)$.
(Possible values are those such that $\tilde f$ is a simplicial map.)
\item Set $\alpha$ to be a uniform random value in $[0,1]$.
\item If $\alpha < \kappa$ or $d(\tilde f, g) \leq d(f_{i-1},g)$, set $f_i
= \tilde{f}$.  Otherwise, goto step (a) above. 
\end{enumerate}
\item Goto step (2) above.
\end{enumerate}
\end{alg}

To distinguish $P$ and $T$, we observe that as $k$ grows, we expect
the number of contiguity classes of based maps from $S^1_k$ to $T$ to
grow polynomially whereas the number of contiguity classes of maps
from $S^1_k$ to $P$ will grow exponentially.  The point is that new
contiguity classes of $T$ are generated by horizontal or vertical
loops, whereas new contiguity classes of $P$ can be described by words
in the ``pinch'' edges.  The expected numerical differences are
sufficiently stark that even estimates of the ranks should expose the
differences.

We use Algorithm~\ref{alg:contclass} to estimate the number of
contiguity classes as follows.  

\begin{alg}\label{alg:countcontclass}
We run the following procedure for a predetermined amount of time.

\begin{enumerate}
\setlength{\itemindent}{-2em}
\item Initialize $L$ to be the empty list.
\item Let $\gamma$ be a uniformly randomly chosen $k$-cycle in the
1-skeleton of the target complex $Y$ that starts and ends at an
arbitrarily chosen base point, representing based simplicial maps from
$S^1_k \to Y$.  
\item If $\gamma$ is not in the same contiguity class as any of the cycles
in $L$, append $\gamma$ to $L$.
\item Goto step (2) above.
\end{enumerate}

\end{alg}

Although it is necessary to choose the stopping time arbitrarily, in
practice we find it most effective to successively increase the
stopping times until the size of $L$ stops changing before the
stopping time is reached.

Selected results of running Algorithm~\ref{alg:countcontclass} on
based maps from $S^1_k$ to $T$ and $P$ are summarized in
Table~\ref{tab}.  We used parameters $M = 500000$ and $\kappa = 0.1$
in Algorithm~\ref{alg:contclass}.  The implementation was in C (with a
Python wrapper) and required about 400 lines of code.  

Scaling by the square of the cycle length, we see that the values for
$T$ are growing roughly quadratically and the values for $P$ are
growing faster (in fact, exponentially).  These results are consistent
with computation by hand of exact values, and clearly indicate a
difference between $P$ and $T$.

\begin{table}
\begin{tabular}{|c|c|c|c|c|}
\hline
\headlower{$k$} & \headlower{$T$} & \headlower{$P$} & $\frac{T}{k^2}$\headstrut & $\frac{P}{k^2}$\headstrut \\
\hline
\entrystrut
9 & 41 & 40 & .506 &\ .494 \\
\entrystrut
12 & 48 & 58 & .333 &\ .403 \\
\entrystrut
15 & 66 & 124 & .293 &\ .551 \\
\entrystrut
18 & 103 & 242 & .318 &\ .747 \\
\entrystrut
21 & 165 & 506 & .374 & 1.143 \\
\hline
\end{tabular}
\vspace{2 pt}
\caption{Estimates of counts of contiguity classes of maps from $S_1^k$.}\label{tab}
\end{table}

Finally, note that for simplicial complexes generated as Rips
complexes of point clouds, determining the contiguity of maps $f$ and
$g$ can be done particularly efficiently even when use of a lookup
table is infeasible.  

\begin{lem}\label{lem:ripscontig}
Let $X$ be a simplicial complex and $(M, \partial_M)$ a finite metric
space.  Simplicial maps $f_0, f_1, \ldots f_n \colon X \to
R_{\epsilon}(M)$ are mutually contiguous if and only if for every
simplex $\{x_0, \ldots, x_m\}$ in $X$, $\partial_M(f_i(x_k),
f_j(x_\ell)) \leq \epsilon$.
\end{lem}

For expository purposes we chose an example which started with
simplicial complexes, but it is straightforward to generate synthetic
point clouds that illustrate the same phenomenon. 

\section{The contiguity complex and persistent homology}\label{sec:pers}

The example in the preceding section involved aggregating information
from a range of mesh sizes on the test space.  Moreover, in practice
it is often necessary to assemble and summarize information from a
range of feature scales on the target space.  Thinking about how to
systematically handle such considerations lead us quickly to ideas of
persistent homology~\cite{EdelsbrunnerLetscherZomorodian}.  In this
section, we briefly sketch the connection of our work to persistence;
a detailed investigation is the subject of future work.

Recall that given a finite metric space $(M, \partial_M)$, for each
$\epsilon > 0$ we can associate the Rips complex $R_\epsilon(M)$.  This
is a simplicial complex with vertices the points of $M$, and higher
simplices specified by the rule that a subset $\{v_0, v_1, \ldots,
v_n\} \subset M$ spans an $n$-simplex when 
\[
\partial_M(v_i, v_j) \leq \epsilon, \,\,\, 0 \leq i < j \leq n.
\]
The construction of the Rips complex is
functorial in both the metric space and in $\epsilon$, with respect to
metric maps and increasing order, respectively.  Specifically, there are 
natural maps
\[
R_\epsilon(M) \to R_{\epsilon'}(M)
\]
for $\epsilon \leq \epsilon'$.  The idea of persistence is that since a priori 
knowledge about the correct value of $\epsilon$ may be hard to come
by, one should study invariants which assemble information as
$\epsilon$ varies.  For any finite metric space
$(M, \partial_M)$, there is a finite collection of values
$\{\epsilon_1, \epsilon_2, \ldots, \epsilon_m\}$ at which the Rips
complex changes.  The persistent homology of 
$(M,\partial_M)$ is a collection of homological invariants of the
direct system of Rips complexes 
\[
R_0(M) \to R_{\epsilon_1}(M) \to R_{\epsilon_2}(M) \to \ldots \to
R_{\epsilon_m}(M).
\]
In fact, persistent homology can be defined for any direct system of
simplicial complexes~\cite{ZomorodianCarlsson}.  In the following, we 
will suppress the choice of coefficients (typically a field).

The contiguity complex is functorial in the target variable: for a
simplicial complex $X$, $\CMap(X,-)$ is a covariant functor to
simplicial complexes.  As a consequence, for $\epsilon \leq \epsilon'$
there is a natural map
\[
\CMap(X, R_\epsilon(M)) \to \CMap(X, R_{\epsilon'}(M)).
\]

\begin{defn}\label{defn:perscont}
Let $(M, \partial_M)$ be a finite metric space (the data) and $Z$ a
simplicial complex (the test space).  The $k$th persistent contiguity
homology of maps from $Z$ to $M$ is the $k$th persistent homology
associated to the direct system of mapping complexes  
\[
\CMap(Z, R_{\epsilon_1}(M)) \to \CMap(Z, R_{\epsilon_2}(M)) \to \ldots \to
\CMap(Z,R_{\epsilon_m}(M)).
\]
\end{defn}

One of the appealing features of the contiguity complex (and the
persistent contiguity homology) is that computation of persistence
becomes particularly computationally tractable when the target is a
Rips complex,  as indicated by Lemma~\ref{lem:ripscontig}.

\begin{example}\label{exa:sparsedisk}
Suppose that $M$ is a point cloud with the property that for
$\epsilon_1 < \epsilon_2$ the Rips complex $R_{\epsilon_1}(M)$ is
homotopy equivalent to a circle and the Rips complex
$R_{\epsilon_2}(M)$ is homotopy equivalent to a point (and before
$\epsilon_1$ is discrete).  Concretely, suppose that $M$ consists of
samples from the unit disk in $\mathbb{R}^2$ such that the density of
the points is
significantly higher on the boundary than on the interior.

Let $S^1_k$ be the simplicial complex modelling the circle that has
$k$ vertices, as in Section~\ref{sec:comp}.  Suppose that $k >
d |M|$ for some number $d$.  Roughly speaking, the barcode associated to the
zeroth persistent contiguity homology will have the following
behavior.  For $\epsilon < \epsilon_1$, there are bars (contiguity
classes) for each vertex in $M$.  When $\epsilon \in (\epsilon_1, \epsilon_2)$,
the bars associated to interior points will remain, but the bars on
the boundary will all collapse to a single point.  However, new bars
will arise for the contiguity classes corresponding to maps of
non-zero degree (up
to at least $d$ in absolute value).  Note that there will be more of
these than the number of homotopy classes, 
as in Example~\ref{exa:simpfundgp}.  When $\epsilon
> \epsilon_2$, all of the bars from the discrete phase will merge into
the single bar from the boundary points and most of the bars
corresponding to the maps on the boundary will merge into a single
bar (but not necessarily all, depending on the precise value of $k$). 
\end{example}

The work of this paper tells us that to study maps out of a test
homotopy type $|X|$ (presented as the geometric realization of a
simplicial complex $X$), the construction of Definition~\ref{defn:perscont} 
should be applied to a subdivision $X'$ of $X$ with a sufficiently
small mesh size.  However, just as it may not be clear how to
understand the correct feature scale $\epsilon$ for the target, a
priori information about suitable mesh sizes for $X'$ may be
unavailable.  For instance, it is clear from
Example~\ref{exa:sparsedisk} that choosing the mesh size (i.e., the
value of $k$) properly depends on intimate knowledge of the behavior
of the data set.  Moreover, the example of the previous section indicates
that the rate of change of the size of the contiguity complex as the
mesh size for $X'$ is itself a very useful invariant.

To this end, we observe that the contiguity complex is also functorial
in the source: for a simplicial complex $Y$, $\CMap(-,Y)$ is a
contravariant functor to simplicial complexes.  Thus, we can consider
the direct system induced by successive subdivision of the target.  Let
$X_n$ be a sequence of successive subdivisions of $X$ and choose
compatible simplicial approximations $X_{n+1} \to X_n$ of the
homeomorphisms $|X_{n+1}| \cong |X_n|$.  

\begin{defn}
Let $X$ and $Y$ be simplicial complexes.  The $k$th persistent
subdivision homology from $X$ to $Y$ relative to subdivisions
$\{X_1,\ldots,X_\ell\}$ is given by the $k$th persistent homology of the
direct system  
\[
\CMap(X,Y) \to \CMap(X_1, Y) \to \ldots \to \CMap(X_{\ell}, Y).
\]
\end{defn}

Different choices of subdivisions $X_n$ and $X_n'$ will lead to
different persistent subdivsion invariants.  However,
Theorem~\ref{thmapprox} and its proof
suggest that these different choices lead to essentially the same
information when the range of mesh sizes is comparable.

We do not do a detailed example here, but note that we can begin to
understand the behavior of the zeroth persistent subdivision homology
from the computations behind the example in Section~\ref{sec:comp}.
In that case, the number of bars increases as $k$ increases, but at
each stage some existing bars merge, as maps in the same homotopy
class but different contiguity classes can enter the same contiguity
class.  This merging can be explicitly tracked using the algorithms we
described above.

Finally, observe that when studying a finite metric space $M$, for a
test space $X$ and subdivisions $X_n$ of $X$, we can assemble the
collection $\{\CMap(X_n, R_\epsilon(M))\}$ as $n$ and $\epsilon$
vary into a diagram suitable for application of multidimensional
persistence~\cite{CarlssonZomorodian}.  We suspect that this
persistent homological invariant is the appropriate mapping invariant
to study in this context.  A detailed study of these persistent
mapping invariants is the subject of future work.

\section{Conclusions}

The classical notions of maps of simplicial complexes and contiguity
of maps of simplicial complexes generalize to define a simplicial
complex of maps between simplicial complexes.  The classical
Simplicial Approximation Theorem says that maps of simplicial
complexes approximate continuous maps between geometric realizations
once subdivisions of small enough mesh are considered.  This theorem
generalizes to mapping spaces to show that the geometric realization
of the simplicial complex of maps provides a good approximation of the
space of maps when subdivisions of small enough mesh are considered.
The number of contiguity classes of maps from small target spaces can
be effectively computed and allows us to distinguish between data sets
that cannot be distinguished by homology alone.  The construction of
the simplicial complex of maps is functorial, and so extends by
naturality to the setting of persistence.

\bibliographystyle{plain}

\end{document}